\documentclass[a4paper,10pt]{article}
\usepackage{geometry}

\usepackage[T1]{fontenc}
\usepackage[utf8]{inputenc}

\usepackage{amssymb, amsthm, amsmath}

\usepackage{authblk}

\usepackage{color}

\usepackage{comment}

\theoremstyle{definition}
\newtheorem{definition}{Definition}
\newtheorem{example}[definition]{Example}
\newtheorem{construction}{Construction}

\theoremstyle{plain}
\newtheorem{theorem}{Theorem}
\newtheorem{proposition}[definition]{Proposition}
\newtheorem{lemma}[definition]{Lemma}
\newtheorem{remark}[definition]{Remark}
\newtheorem{corollary}[definition]{Corollary}

\newcommand{\w}{{\rm w}}
\newcommand{\dd}{{\rm d}}
\newcommand{\C}{\mathcal{C}}

\title{Linear codes in the folded Hamming distance and\\the quasi MDS property}
\author{Umberto Mart{\'i}nez-Pe\~{n}as\thanks{umberto.martinez@uva.es} }
\author{Rub{\'e}n Rodr{\'i}guez-Ballesteros\thanks{ruben.rodriguez22@estudiantes.uva.es; rubenrb@iesgalileo.es}}
\affil{IMUVa-Mathematics Research Institute,\\University of Valladolid, Spain}

\date{}

\begin{document}

\maketitle

\begin{abstract}
In this work, we study linear codes with the folded Hamming distance, or equivalently, codes with the classical Hamming distance that are linear over a subfield. This includes additive codes. We study MDS codes in this setting and define quasi MDS (QMDS) codes and dually QMDS codes, which attain a more relaxed variant of the classical Singleton bound. We provide several general results concerning these codes, including restriction, shortening, weight distributions, existence, density, geometric description and bounds on their lengths relative to their field sizes. We provide explicit examples and a binary construction with optimal lengths relative to their field sizes, which beats any MDS code. 


\end{abstract}

\section{Introduction} \label{sec intro}

In this manuscript, we study linear codes in the folded Hamming distance or, equivalently, codes in the classical Hamming distance which are linear over a subfield (this includes additive codes). 

We define quasi MDS codes (or QMDS codes), which lie in between classical MDS codes and almost MDS codes \cite{faldum} (their dimensions are larger than those of almost MDS codes for a given minimum distance). We show that their duals are not QMDS in general and then define dually QMDS codes (both themselves and their duals are QMDS), which lie between classical MDS codes and near MDS codes \cite{dodunekov2}. As we show in Section \ref{sec long qmds}, there exist dually QMDS codes whose lengths relative to their field size beat the MDS conjecture \cite{ball}. This makes them an interesting family of codes to study, since they have better parameters than almost and near MDS codes and at the same time can be longer than any classical MDS code. 

Linear codes in the folded Hamming distance have been studied in the context of byte error correction \cite{etzion}, low density MDS codes \cite{blaum-lowest, xu} and recently in relation to quantum codes \cite{ball-additive}. Notice also that linear codes in the folded Hamming distance can be seen as array or matrix codes with the Hamming distance defined column-wise \cite{blaum-lowest, xu}. In particular, they are a special case of linear codes in the sum-rank distance \cite{fnt}. However, most properties and constructions discussed in this manuscript simply do not hold for the sum-rank distance in general \cite{alberto-fundamental}.

Interestingly, most capacity-achieving and efficient list-decodable codes turn out to be QMDS or dually QMDS codes in the folded Hamming distance. Such codes fall under the umbrella family of polynomial ideal codes \cite{bhandari, noga}, which include in particular folded Reed--Solomon codes and (univariate) multiplicity codes (see \cite{bhandari}). We show their dually QMDS property in Section \ref{sec density}.

The contributions of this manuscript are as follows. In Section \ref{sec block distance}, we provide some preliminary definitions and results on the folded Hamming distance, mainly concerning duality and code equivalence. In Section \ref{sec qmds codes}, we introduce QMDS and dually QMDS codes, characterize their minimum distances and dimensions and study their restricted and shortened codes. In Section \ref{sec density}, we explicitly show the existence of dually QMDS for all parameters for large field sizes (using polynomial ideal codes) and then show that the family of dually QMDS codes is dense. In Section \ref{sec weight dist}, we study their weight distributions, which make use of the previous results on restriction and shortening, and which will be later used for bounds on code lengths. In Section \ref{sec long qmds}, we study how long (dually) QMDS codes can be. We provide two upper bounds on the code length relative to their field sizes, some examples of QMDS codes and a general binary construction of dually QMDS codes with optimal lengths, longer than any MDS code. Finally, in Section \ref{sec pseudo arcs}, we provide a 1-1 correspondence between equivalence classes of linear codes in the folded Hamming distance and equivalence classes of pseudo arcs, which are useful for constructing MSRD codes and PMDS codes \cite{liu, generalMSRD}.

\section{The folded Hamming distance} \label{sec block distance}

In this section, we introduce the folded Hamming distance, define duality and characterize its linear isometries. In the following, $ \mathbb{F}_q $ denotes the finite field with $ q $ elements. We will also denote $ [n] = \{ 1,2, \ldots, n \} $ and $ [m,n] = \{m,m+1, \ldots, n \} $ for integers $ m \leq n $.

\begin{definition}
For $\mathbf{c} = (\mathbf{c}_1,\ldots,\mathbf{c}_n) \in \mathbb{F}_{q}^{rn} $,  where $ \mathbf{c}_i \in \mathbb{F}_q^r $ for $ i \in [n] $, its folded Hamming weight is defined as $\w_F(\mathbf{c})=\lvert\{i\in[n]:\mathbf{c}_i\neq0\}\rvert$. We define the folded Hamming distance between $ \mathbf{c},\mathbf{d} \in \mathbb{F}_q^{rn} $ as $ \dd_F(\mathbf{c},\mathbf{d}) = \w_F(\mathbf{c}- \mathbf{d}) $. In general, a code is a subset $ \mathcal{C} \subseteq \mathbb{F}_q^{rn} $. We define the minimum folded Hamming distance of $ \mathcal{C} $ as $ \dd(\mathcal{C}) = \min \{ \dd_F(\mathbf{c},\mathbf{d}) : \mathbf{c},\mathbf{d} \in \mathcal{C}, \mathbf{c} \neq \mathbf{d} \} $. 

If $ \mathcal{C} \subseteq \mathbb{F}_q^{rn} $ is $ \mathbb{F}_q $-linear, notice that $ d = \dd(\mathcal{C}) = \min \{ \w_F(\mathbf{c}) : \mathbf{c} \in \mathcal{C} \setminus \{ 0 \} \} $. If $ k $ is its dimension, then we say that $ \mathcal{C} $ is a code of type $ [n,r,k,d] $.
\end{definition}

Clearly, $ \mathbb{F}_q $-linear codes in $ \mathbb{F}_q^{rn} $ with the folded Hamming distance are the same as $ \mathbb{F}_q $-linear codes in $ \mathbb{F}_{q^r}^n $ with the classical Hamming distance, due to the following. Let $ \boldsymbol\beta = (\beta_1, \ldots, \beta_r) \in \mathbb{F}_{q^r}^r $ be an ordered basis of $ \mathbb{F}_{q^r} $ over $ \mathbb{F}_q $. Define the expansion map $ \varepsilon_{\boldsymbol\beta} : \mathbb{F}_{q^r} \longrightarrow \mathbb{F}_q^r $ by $ \varepsilon_{\boldsymbol\beta} (c_1 \beta_1 + \cdots + c_r \beta_r) = (c_1, \ldots, c_r) $, for $ c_1, \ldots, c_r \in \mathbb{F}_q $. If we extend it componentwise, it is obvious that $ \varepsilon_{\boldsymbol\beta} : \mathbb{F}_{q^r}^n \longrightarrow \mathbb{F}_q^{rn} $ is an $ \mathbb{F}_q $-linear isometry considering the classical Hamming distance in $ \mathbb{F}_{q^r}^n $ and the folded Hamming distance in $ \mathbb{F}_q^{rn} $. Note, however, that over infinite fields (such as $ \mathbb{R} $) there may be no field extension of degree $ r $ for every positive integer $ r $. 

Notice that additive codes with the classical Hamming distance in $ \mathbb{F}_q^n $, where $ q = p^r $ and $ p $ is prime, are thus equivalent to $ \mathbb{F}_p $-linear codes with the folded Hamming distance in $ \mathbb{F}_p^{rn} $.

We will focus on duality based on the usual inner product in $ \mathbb{F}_q^{rn} $, given as follows.

\begin{definition}
We define the inner product between $ \mathbf{c} , \mathbf{d} \in \mathbb{F}_q^{rn} $ as $ \mathbf{c} \cdot \mathbf{d} = c_1d_1 + \cdots + c_{rn} d_{rn} $, where $ \mathbf{c} = (c_1, \ldots, c_{rn}) $ and $ \mathbf{d} = (d_1, \ldots, d_{rn}) $. Given an $ \mathbb{F}_q $-linear code $ \mathcal{C} \subseteq \mathbb{F}_q^{rn} $, we define its dual as $ \mathcal{C}^\perp = \{ \mathbf{d} \in \mathbb{F}_q^{rn} : \mathbf{c} \cdot \mathbf{d} = 0, \textrm{ for all } \mathbf{c} \in \mathcal{C} \} $.
\end{definition}

Notice that the usual inner product in $ \mathbb{F}_q^{rn} $ is not $ \mathbb{F}_{q^r} $-bilinear when considered in $ \mathbb{F}_{q^r}^n $ via the maps $ \varepsilon_{\boldsymbol\beta} $. However, when considering $ \mathbb{F}_{q^r} $-linear codes, their duals with respect to the usual inner products in $ \mathbb{F}_{q^r}^n $ and $ \mathbb{F}_q^{rn} $ coincide if we use appropriate expansion maps. 

\begin{proposition}
Let $ \mathcal{C} \subseteq \mathbb{F}_{q^r}^n $ be an $ \mathbb{F}_{q^r} $-linear code and let $ \mathcal{C}^\perp \subseteq \mathbb{F}_{q^r}^n $ denote its dual with respect to the usual inner product in $ \mathbb{F}_{q^r}^n $. If $ \boldsymbol\beta = (\beta_1, \ldots, \beta_r) $ and $ \boldsymbol\alpha = (\alpha_1, \ldots, \alpha_r) $ are dual ordered bases of $ \mathbb{F}_{q^r} $ over $ \mathbb{F}_q $ (i.e., $ {\rm Tr}(\beta_i\alpha_j) = \delta_{i,j} $, where $ {\rm Tr} $ is the trace of $ \mathbb{F}_{q^r} $ over $ \mathbb{F}_q $), then 
$$ \varepsilon_{\boldsymbol\beta}\left( \mathcal{C}^\perp \right) = \varepsilon_{\boldsymbol\alpha}(\mathcal{C})^\perp. $$
(Dual bases of $ \mathbb{F}_{q^r} $ over $ \mathbb{F}_q $ always exist, see \cite[p. 54]{lidl}.)
\end{proposition}
\begin{proof}
By counting dimensions over $ \mathbb{F}_q $, we only need to show that $ \varepsilon_{\boldsymbol\beta}\left( \mathcal{C}^\perp \right) \subseteq \varepsilon_{\boldsymbol\alpha}(\mathcal{C})^\perp $. Now this holds since, given $ \mathbf{c} = (c_1,\ldots,c_n) \in \mathcal{C} $ and $ \mathbf{d} = (d_1, \ldots, d_n) \in \mathcal{C}^\perp $, with $ \varepsilon_{\boldsymbol\beta}(c_i) = (c_{i,1}, \ldots, c_{i,r}) $ and $ \varepsilon_{\boldsymbol\alpha}(d_i) = (d_{i,1}, \ldots, d_{i,r}) $, for $ i \in [n] $, we have that
$$ 0 = {\rm Tr}(0) = {\rm Tr}\left( \sum_{i=1}^n c_i d_i \right) = \sum_{i=1}^n \sum_{j=1}^r \sum_{k=1}^r c_{i,j}d_{i,k} {\rm Tr}(\beta_j \alpha_k) = \sum_{i=1}^n \sum_{j=1}^r c_{i,j}d_{i,j}. $$ 
\end{proof}

Finally, we notice that the folded Hamming distance in $ \mathbb{F}_q^{rn} $ coincides with the sum-rank distance in such a space by considering a vector in $ \mathbb{F}_q^{rn} $ as a tuple of $ n $ matrices over $ \mathbb{F}_q $ of size $ 1 \times r $. See \cite{fnt} for more information on the sum-rank distance. In particular, the $ \mathbb{F}_q $-linear isometries for the folded Hamming distance are known, see \cite[Th. 2]{sr-hamming}.

\begin{proposition} [\cite{sr-hamming}] \label{prop isometries}
Let $ \phi : \mathbb{F}_q^{rn} \longrightarrow \mathbb{F}_q^{rn} $ be an $ \mathbb{F}_q $-linear vector space isomorphism. Then $ \w_F(\phi(\mathbf{c})) = \w_F(\mathbf{c}) $, for all $ \mathbf{c} \in \mathbb{F}_q^{rn} $, if and only if there exist invertible matrices $ A_1, \ldots, A_n \in {\rm GL}_r(\mathbb{F}_q) $ and a permutation $ \sigma : [n] \longrightarrow [n] $ such that, for all $ \mathbf{c}_1, \ldots, \mathbf{c}_n \in \mathbb{F}_q^{r} $,
$$ \phi(\mathbf{c}_1, \ldots, \mathbf{c}_n) = \left( \mathbf{c}_{\sigma(1)} A_1, \ldots, \mathbf{c}_{\sigma(n)}A_n \right). $$
\end{proposition}

\begin{definition} \label{def equivalence}
We say that two linear codes $ \mathcal{C},\mathcal{C}^\prime \subseteq \mathbb{F}_q^{rn} $ are equivalent if there is an $ \mathbb{F}_q $-linear isometry for the folded Hamming distance $ \phi : \mathbb{F}_q^{rn} \longrightarrow \mathbb{F}_q^{rn} $ such that $ \mathcal{C}^\prime = \phi(\mathcal{C}) $.
\end{definition}

The following result is straightforward from Proposition \ref{prop isometries}.

\begin{corollary} \label{cor dual equivalence}
Two linear codes are equivalent if, and only if, so are their duals.
\end{corollary}

However, note that the results that we will obtain in this manuscript are either unkown or simply do not hold for the sum-rank distance in general. For instance, MacWilliams equations (Theorem \ref{mw}) do not exist in general for the sum-rank distance \cite{alberto-fundamental}.

Throughout the remainder of the manuscript, linear will mean $ \mathbb{F}_q $-linear.

\section{Quasi MDS codes} \label{sec qmds codes}

In this section, we provide a relaxed version of the Singleton bound and define and study QMDS and dually QMDS codes.

Item 1 in the following proposition follows from the classical Singleton bound for linear or non-linear codes in the classical Hamming distance \cite[Th. 2.4.1]{pless}. Item 2 is straightforward from Item 1, and Item 3 is the dual statement.

\begin{proposition}\label{singleton}
Let $ \mathcal{C} $ be a linear code of type $ [n,r,k,d] $ and let $ \mathcal{C}^\perp $ be its dual, of type $ [n,r,rn-k,d^\perp] $. Then 
\begin{enumerate}
\item \label{one} $k\leq r(n - d +1)$, 
\item \label{two} $d\leq n-\lceil \frac{k}{r} \rceil+1=n-\lfloor \frac{k-1}{r} \rfloor$, and
\item \label{three} $d^\perp\leq \lfloor \frac{k}{r} \rfloor+1 = \lceil \frac{k+1}{r} \rceil $.
\end{enumerate}
\end{proposition}

As usual, a linear code is MDS if it attains the bound in Item 1. If $ r \mid k $ (necessary for the code to be MDS), then Items 1 and 2 coincide. However, when $ r \nmid k $, the second bound may be attained but the first one cannot. This motivates the following definition.

\begin{definition}
We say that a linear code of type $[n,r,k,d]$ is quasi-MDS or QMDS if $d=n-\lceil \frac{k}{r} \rceil+1=n-\lfloor \frac{k-1}{r} \rfloor$. A linear MDS code is a QMDS code such that $ r \mid k $.
\end{definition}

The dual of a linear MDS code is MDS, also for the folded Hamming distance. The following result has been independently proven in \cite[Lemma 3.3]{blaum-lowest}, \cite[Th. 1]{xu} and \cite[Th. 9]{ball-additive}.

\begin{proposition} [\cite{ball-additive, blaum-lowest, xu}] \label{prop dual of MDS}
A linear code $ \mathcal{C} \subseteq \mathbb{F}_q^{rn} $ is MDS if, and only if, so is its dual.
\end{proposition}

However, this is not the case for QMDS codes in general.

\begin{example} \label{example qmds 1}
Over any field, the linear code of type $ [3,3,4,2] $ with the following generator matrix is QMDS but its dual is of type $ [3,3,5,1] $, thus not QMDS:
$$ G = \left( \begin{array}{ccc|ccc|ccc}
1 & 0 & 0 & 1 & 0 & 0 & 0 & 0 & 0 \\
0 & 1 & 0 & 0 & 1 & 0 & 0 & 0 & 0 \\
0 & 0 & 1 & 0 & 0 & 0 & 0 & 1 & 0 \\
1 & 0 & 0 & 0 & 1 & 0 & 1 & 0 & 0 
\end{array} \right). $$
\end{example}

This motivates the following definition (we will give examples in Sections \ref{sec density} and \ref{sec long qmds}).

\begin{definition}
A linear code is dually QMDS if both itself and its dual are QMDS.
\end{definition}

\begin{remark}
A $ q $-analog of QMDS and dually QMDS codes have been considered before \cite{gorla}, called quasi maximum rank distance (QMRD) and dually QMRD codes. However, in the rank metric there exist linear MRD codes for any choice of parameters \cite{gabidulin}. In contrast, QMDS and dually QMDS codes may attain lengths and finite-field sizes not achievable by MDS codes as we show in Section \ref{sec long qmds}. Other properties make the folded-Hamming-metric counterpart different in essence, such as the density results (Section \ref{sec density}) or their applications in practice \cite{ball-additive, bhandari, blaum-lowest}.
\end{remark}

A first observation is that the dually QMDS property is preserved by equivalence, which follows from Corollary \ref{cor dual equivalence}.

\begin{proposition}
A code that is equivalent to a dually QMDS code is also dually QMDS.
\end{proposition}

We may also characterize dually QMDS codes in terms of the sum of the distances of the code and its dual.

\begin{proposition} \label{prop d+d dual}
Let $ \mathcal{C} $ be a linear code of type $ [n,r,k,d] $ and let $ \mathcal{C}^\perp $ be its dual, of type $ [n,r,rn-k,d^\perp] $. If $ r \mid k $, then either $ d+d^\perp = n+2 $ or $ d+d^\perp \leq n $, whereas if $ r \nmid k $, then $ d+d^\perp \leq n+1 $. Furthermore, the following hold:
\begin{enumerate}
\item
$ \mathcal{C} $ is MDS if, and only if, $ \mathcal{C}^\perp $ is MDS if, and only if, $ d+d^\perp = n+2 $.
\item
$ \mathcal{C} $ is dually QMDS and $ r \nmid k $ if, and only if, $ d+d^\perp = n+1 $.
\end{enumerate}
\end{proposition}
\begin{proof}
If we add the Singleton bounds for $ \mathcal{C} $ and $ \mathcal{C}^\perp $ (Items 2 and 3 in Proposition \ref{singleton}), we obtain 
$$ d+d^\perp \leq \left( n - \left\lceil \frac{k}{r} \right\rceil +1 \right) + \left( \left\lfloor \frac{k}{r} \right\rfloor + 1 \right) \leq n+2 . $$
Notice however that, if $ r \nmid k $, then $ \lceil \frac{k}{r} \rceil = \lfloor \frac{k}{r} \rfloor +1 $, which implies in this case that $ d+d^\perp \leq n+1 $. Note also that, if $ r \mid k $ and $ d+d^\perp = n+1 $, then $ \mathcal{C} $ is MDS but $ \mathcal{C}^\perp $ is not (or viceversa), which is not possible by Proposition \ref{prop dual of MDS}. Thus if $ r \mid k $ and $ d+d^\perp < n+2 $, then $ d+d^\perp \leq n $.

In particular, $ d+d^\perp = n+2 $ may only happen if $ r \mid k $, in which case it must also hold that $ k = r(n-d+1) $ and $ rn-k = r(n-d^\perp+1) $, and $ \mathcal{C} $ (thus $ \mathcal{C}^\perp $) is MDS. Conversely, if $ \mathcal{C} $ (thus $ \mathcal{C}^\perp $) is MDS, then we have $ d+d^\perp = n+2 $, and Item 1 is proven.

Similarly, if $ r \nmid k $ and $ d+d^\perp = n+1 $, then equalities in both Items 2 and 3 in Proposition \ref{singleton} must hold, i.e., $ \mathcal{C} $ is dually QMDS. The reversed implication is straightforward and Item 2 is proven.
\end{proof}

As in the classical case of MDS codes, we may characterize QMDS codes in terms of their generator and parity-check matrices. We start with the following result, which generalizes \cite[Cor. 1.4.14 \& Th. 1.4.15]{pless} from $ r = 1 $ to $ r \geq 1 $ in general. For a $ k \times (rn) $ matrix $ G = (G_1|\ldots|G_n) $, where $ G_i $ is of size $ k \times r $, we say that $ G_i $ is the $ i $th column block of $ G $ formed by $ r $ columns.

\begin{proposition} \label{prop d from generator and parity matrices}
Let $ \mathcal{C} $ be a linear code of type $ [n,r,k,d] $ with generator matrix $ G $ and parity-check matrix $ H $.
\begin{enumerate}
\item
$ d $ is the maximum number such that the submatrix formed by the $ r(n-d+1) $ columns of any set of $ n-d+1 $ column blocks of $ G $ has rank $ k $.
\item
$ d-1 $ is the maximum number such that the submatrix formed by the $ r(d-1) $ columns of any set of $ d-1 $ column blocks of $ H $ has rank $ r(d-1) $.
\end{enumerate}
\end{proposition}
\begin{proof}
For Item 1, $ d $ is the minimum number such that it is possible to obtain zeros in $ n-d $ column blocks by making (nontrivial) linear combinations with the rows of $ G $. As a consequence, any submatrix formed by the columns of $ G $ in $ n-d+1 $ column blocks must have rank $ k $ since $ k \leq r(n-d+1) $, and there exists $ n-d $ column blocks in $ G $ whose $ r(n-d) $ columns form a matrix of rank lower than $ k $.

For Item 2, since $ HG^\intercal = 0 $, any set of $ n-d+1 $ column blocks of $ G $ has maximum rank if, and only if, it is not possible to obtain zeros in $ n - (n-d+1) = d-1 $ column blocks in nontrivial linear combinations of the rows of $ H $. In other words, any submatrix formed by the columns of $ H $ in $ d-1 $ column blocks must have rank $ r(d-1) $ since $ r(d-1) \leq rn - k $.
\end{proof}

As a consequence, we obtain the following characterizations of QMDS codes, which recovers the characterizations of classical MDS codes \cite[Th. 2.4.3]{pless} when $ r = 1 $.

\begin{corollary} \label{cor gen matrices QMDS}
Let $ \mathcal{C} $ be a linear code of type $ [n,r,k,d] $ with generator matrix $ G $ and parity-check matrix $ H $. The following are equivalent:
\begin{enumerate}
\item
$ \mathcal{C} $ is QMDS.
\item
Any $ k \times r \lceil \frac{k}{r} \rceil $ submatrix of $ G $ formed by $ \lceil \frac{k}{r} \rceil $ column blocks has rank $ k $.
\item
Any $ (rn-k) \times r (n - \lceil \frac{k}{r} \rceil) $ submatrix of $ H $ formed by $ n - \lceil \frac{k}{r} \rceil $ column blocks has rank $ r (n - \lceil \frac{k}{r} \rceil) $. 
\end{enumerate}
In particular, $ \mathcal{C} $ is dually QMDS if, and only if, any $ k \times r \lceil \frac{k}{r} \rceil $ submatrix of $G$ formed by $ \lceil \frac{k}{r} \rceil $ column blocks has rank $ k $ and any $ k \times r \lfloor \frac{k}{r} \rfloor $ submatrix of $G$ formed by $ \lfloor \frac{k}{r} \rfloor $ column blocks has rank $ r \lfloor \frac{k}{r} \rfloor $.
\end{corollary}

\begin{remark}
Note that Proposition \ref{prop dual of MDS} follows immediately from the previous corollary.
\end{remark}

We now turn to restricted and shortened codes. Apart from being of interest on their own, we will use them when computing the weight distribution of dually QMDS codes in Section \ref{sec weight dist}.

\begin{definition}
Given $ I \subseteq [n] $, we define the projection $ \pi_I : \mathbb{F}_q^{rn} \longrightarrow \mathbb{F}_q^{r|I|} $ such that $ \pi_I(\mathbf{c}_1, \ldots, \mathbf{c}_n) $ $ = (\mathbf{c}_i)_{i \in I} $, for $ \mathbf{c}_1, \ldots, \mathbf{c}_n \in \mathbb{F}_q^r $.
\end{definition} 

\begin{definition}
Given a code $ \mathcal{C} \subseteq \mathbb{F}_q^{rn} $ and $ I \subseteq [n] $, we define the restricted code $ \mathcal{C}^I = \pi_I(\mathcal{C}) $, the null subcode $ \mathcal{C}(I) = \mathcal{C} \cap \ker(\pi_I) $ and the shortened code $ \mathcal{C}_I = \pi_I(\mathcal{C}([n]\setminus I)) $.
\end{definition}

Restricting and shortening QMDS codes yield again QMDS codes. 

\begin{proposition} \label{prop restricted QMDS}
Let $ \mathcal{C} $ be a linear code of type $ [n,r,k,d] $.
\begin{enumerate}
\item
$ \mathcal{C} $ is QMDS if, and only if, $ \mathcal{C}^I $ is of dimension $ k $ for all $ I \subseteq [n] $ with $ r |I| \geq k $.
\item
If $ \mathcal{C} $ is QMDS, then $ \mathcal{C}^I $ is of dimension $ k^I $ with $ r(|I|-1) < k^I \leq r|I| $ for all $ I \subseteq [n] $ with $ r |I| < k $.
\item
$ \mathcal{C} $ is QMDS if, and only if, $ \mathcal{C}_I = 0 $ for all $ I \subseteq [n] $ with $ r(n-|I|) \geq k $.
\item
If $ \mathcal{C} $ is QMDS, then $ \mathcal{C}_I $ is of dimension $ k_I $ with $ k - r(n-|I|) \leq k_I < k - r(n-|I|)+r $ for all $ I \subseteq [n] $ with $ r(n-|I|) < k $.
\end{enumerate}
Furthermore, if $ \mathcal{C} $ is QMDS, then so are $ \mathcal{C}^I $ and $ \mathcal{C}_I $, for all $ I \subseteq [n] $. 
\end{proposition}
\begin{proof}
Item 1 (including the fact that $ \mathcal{C}^I $ is QMDS) follows directly from Corollary \ref{cor gen matrices QMDS}. 

For Item 3, we have $ d > n - \lceil \frac{k}{r} \rceil $ if, and only if, $ \mathcal{C}_I = 0 $ for all $ I \subseteq [n] $ with $ |I| \leq n - \lceil \frac{k}{r} \rceil $, which is equivalent to $ r(n-|I|) \geq k $. 

Next we prove Item 2. Let $ I \subseteq [n] $ with $ r|I| < k $. Consider a generator matrix $ G $ of $ \mathcal{C} $, and let $ G_I $ be the submatrix of $ G $ formed by the column blocks indexed by $ I $. Since $ G_I $ is of size $ k \times (r|I|) $ and $ r|I| < k $, there can be at most $ r-1 $ columns in $ G_I $ that linearly depend on the rest of the columns of $ G_I $ (otherwise there is a $ k \times r \lceil \frac{k}{r} \rceil $ submatrix of $ G $ of rank smaller than $ k $, contradicting Corollary \ref{cor gen matrices QMDS}). Hence $ k^I = \dim(\mathcal{C}^I) \geq r|I|-(r-1) > r(|I|-1) $ and any nonzero linear code in $ \mathbb{F}_q^{r|I|} $ of such a dimension is QMDS.

Finally we prove Item 4. Let $ I \subseteq [n] $ with $ r(n-|I|) < k $. We have $ k-k_I = \dim(\pi_{[n] \setminus I}(\mathcal{C}))  \leq r(n-|I|) $. Thus  
$$ \dd(\mathcal{C}_I) \geq d = n - \left\lceil \frac{k}{r} \right\rceil + 1 \geq n - \left\lceil \frac{r(n-|I|) + k_I}{r} \right\rceil + 1 = |I| - \left\lceil \frac{k_I}{r} \right\rceil + 1. $$
By the Singleton bound, $ \mathcal{C}_I $ is QMDS and the inequalities above are equalities. In particular, $ k_I + r(n-|I|) \leq r \lceil \frac{k}{r} \rceil < k +r $, hence $ k_I < k - r(n-|I|)+r $.
\end{proof}

Similarly, any restriction and shortening of a dually QMDS code is again a dually QMDS code. Moreover, their dimensions characterize the dually QMDS property.

\begin{theorem} \label{th restricted dually qmds}
For a linear code $ \mathcal{C} $ of type $ [n,r,k,d] $, the following are equivalent:
\begin{enumerate}
\item
$ \mathcal{C} $ is dually QMDS.
\item
$ \dim(\mathcal{C}^I) = k $ if $ r |I| \geq k $ and $ \dim(\mathcal{C}^I) = r|I| $ if $ r |I| < k $, for all $ I \subseteq [n] $. 
\item
$ \dim(\mathcal{C}_I) = k - r(n-|I|) $ if $ r (n-|I|) \leq k $ and $ \dim(\mathcal{C}_I) = 0 $ if $r (n-|I|) > k $, for all $ I \subseteq [n] $. 
\end{enumerate}
Furthermore, if $ \mathcal{C} $ is dually QMDS, then so are $ \mathcal{C}^I $ and $ \mathcal{C}_I $, for all $ I \subseteq [n] $.
\end{theorem}
\begin{proof}
We first show that Item 2 implies Item 1. We will repeatedly use $ (\mathcal{C}^I)^\perp = (\mathcal{C}^\perp)_I $ (see \cite[Th. 1.5.7]{pless}). First, from $ \dim(\mathcal{C}^I) = k $ if $ r |I| \geq k $, we deduce that $ \mathcal{C} $ is QMDS by Proposition \ref{prop restricted QMDS}. Next, if $ |I| \leq \lfloor \frac{k}{r} \rfloor $, then
$$ \dim((\mathcal{C}^\perp)_I) = \dim((\mathcal{C}^I)^\perp) = r|I| - \dim(\mathcal{C}^I) = r|I|-r|I| = 0, $$
which means that $ \dd(\mathcal{C}^\perp) \geq \lfloor \frac{k}{r} \rfloor +1 $, and we conclude that $ \mathcal{C} $ is dually QMDS.

Using $ (\mathcal{C}^I)^\perp = (\mathcal{C}^\perp)_I $, one can similarly show that Item 3 implies Item 1. Therefore, we only need to prove that Item 1 implies Items 2 and 3, and that $ \mathcal{C}^I $ and $ \mathcal{C}_I $ are dually QMDS.

Assume that $ \mathcal{C} $ is dually QMDS. First, $ \mathcal{C}^I $ and $ \mathcal{C}_I $ are QMDS by Proposition \ref{prop restricted QMDS}. Using that $ (\mathcal{C}^I)^\perp = (\mathcal{C}^\perp)_I $, $ (\mathcal{C}_I)^\perp = (\mathcal{C}^\perp)^I $ and that $ \mathcal{C}^\perp $ is also QMDS, we deduce that both $ \mathcal{C}^I $ and $ \mathcal{C}_I $ are dually QMDS, again by Proposition \ref{prop restricted QMDS}.

We now compute dimensions. If $ r |I| \geq k $, then $ \dim(\mathcal{C}^I) = k $ by Proposition \ref{prop restricted QMDS}. Now assume that $ r |I| < k $. Since $ \mathcal{C}^\perp $ is QMDS and $ r(n-|I|) > \dim(\mathcal{C}^\perp) $, then $ \dim((\mathcal{C}^\perp)_I) = 0 $ by Proposition \ref{prop restricted QMDS}. Using again $ (\mathcal{C}^I)^\perp = (\mathcal{C}^\perp)_I $, we have
$$ \dim(\mathcal{C}^I) = r |I| - \dim((\mathcal{C}^I)^\perp) = r|I| - \dim((\mathcal{C}^\perp)_I) = r |I|. $$

We now turn to $ \mathcal{C}_I $. First, using once again $ (\mathcal{C}^I)^\perp = (\mathcal{C}^\perp)_I $, we have that 
$$ \dim(\mathcal{C}_I) = \dim(((\mathcal{C}^\perp)^I)^\perp) = r|I| - \dim((\mathcal{C}^\perp)^I). $$
Using the formula for $ \dim((\mathcal{C}^\perp)^I) $ already computed (since $ \mathcal{C}^\perp $ is dually QMDS), we obtain the formula for $ \dim(\mathcal{C}_I) $ given in the proposition. 
\end{proof}

We now illustrate how the dimension formulas in Theorem \ref{th restricted dually qmds} do not hold for all QMDS codes. 

\begin{example} \label{ex shortening not QMDS}
Consider the QMDS code $ \mathcal{C} $ of type $ [3,3,4,2] $ from Example \ref{example qmds 1}. If $ I = \{ 3 \} $, then $ \dim(\mathcal{C}^I) = 2 \neq 3 $ and if $ I = \{ 1,2 \} $, then $ \dim(\mathcal{C}_I) = 2 \neq 1 $. 
\end{example}

\section{Existence and density of dually QMDS codes} \label{sec density}

In this section we show that dually QMDS codes exist for all parameters given a sufficiently large finite field. We first give an existential result, which also shows that the family of such codes is dense (they appear with probability approaching $ 1 $ for large fields). Then we show that the explicit codes known as polynomial ideal codes \cite{bhandari, noga} are dually QMDS covering general parameters for sufficiently large finite fields. Later in Section \ref{sec long qmds}, we explore dually QMDS over small finite fields relative to their code length.

For the existence and density, we will make use of the DeMillo-Lipton-Schwartz-Zippel bound \cite[Lemma 16.2]{jukna}. We denote by $\mathbb{F}_q[x_1, x_2, \dots, x_m]$ the polynomial ring in $m$ variables $ x_1, x_2, $ $ \ldots, x_m $ over $\mathbb{F}_q$. For $F\in\mathbb{F}_q[x_1, x_2, \dots, x_m]$, we consider $\deg(F)$ as its total degree and we denote $Z(F) = \{ \mathbf{a} \in \mathbb{F}_q^m : F(\mathbf{a}) = 0 \} $.

\begin{lemma}[\textbf{DeMillo-Lipton-Schwartz-Zippel} \cite{jukna}]\label{sz}
  Let $F \in \mathbb{F}_q[x_1, x_2, \dots, x_m]$. Then $$ |Z(F)| \leq \deg(F) \cdot q^{m-1} .$$
\end{lemma}

\begin{proposition}
    Let $n,k,r $ be positive integers and let $\mathbb{F}_q$ be a finite field such that $ k \leq rn $ and
  \[
  C(n,r,k) := k \binom{n}{\lceil \frac{k}{r} \rceil} \binom{r \lceil \frac{k}{r} \rceil}{k} + r \left\lfloor \frac{k}{r} \right\rfloor \binom{n}{\lfloor \frac{k}{r} \rfloor} \binom{k}{r \lfloor \frac{k}{r} \rfloor} < q.
  \]
  \begin{enumerate}
    \item There exists at least one dually QMDS code $\mathcal{C}$ of type $[n, r, k, d]$, $ d = n - \lceil \frac{k}{r} \rceil + 1 $.
    \item The probability that a code $\mathcal{C}$ of type $[n, r, k, d]$ (where $d$ is not fixed) chosen uniformly at random is dually QMDS is at least $1 - \frac{C(n,r,k)}{q}$.
  \end{enumerate}
\end{proposition}
\label{22}
\begin{proof}
We only need to prove Item 2. Choosing a code of type $ [n,r,k,d] $ (where $ d $ is not fixed) uniformly at random is equivalent to choosing a full-rank $ k \times (rn) $ matrix $ G $ over $ \mathbb{F}_q $ (modulo multiplying on the left by invertible $ k \times k $ matrices over $ \mathbb{F}_q $). We consider the entries of $ G $ as variables $ x_1, x_2, \ldots, x_m $ with $ m = krn $, and an instantiation of $ G $ consists in choosing a point in $ \mathbb{F}_q^m $ uniformly at random and evaluating the variables $ x_1, x_2, \ldots, x_m $ in such a point.  

By Corollary \ref{cor gen matrices QMDS}, $G$ generates a dually QMDS code if any $ k \times r \lceil \frac{k}{r} \rceil $ submatrix of $G$ formed by $ \lceil \frac{k}{r} \rceil $ column blocks has rank $ k $ (QMDS condition), and any $ k \times r \lfloor \frac{k}{r} \rfloor $ submatrix of $G$ formed by $ \lfloor \frac{k}{r} \rfloor $ column blocks has rank $ r \lfloor \frac{k}{r} \rfloor $ (dual QMDS condition). This condition holds if all of the involved minors are nonzero, i.e., the product of all such minors is nonzero. Such a product is a polynomial $ F \in \mathbb{F}_q[x_1,x_2, \ldots, x_m] $ of degree $ \deg(F) = C(n,r,k) < q $. The probability that a point in $ \mathbb{F}_q^m $, chosen uniformly at random, does not lie in $ Z(F) $ is
$$ \frac{q^m - |Z(F)|}{q^m} \geq \frac{q^m - \deg(F) q^{m-1}}{q^m} = 1 - \frac{C(n,r,k)}{q}, $$
by Lemma \ref{sz}, and we are done.
\end{proof}

Noting that $ C(n,r,k) $ only depends on $ n,r,k $ and not on $ q $, we conclude the following.

\begin{corollary} \label{cor density}
Dually QMDS codes are dense in the set of codes of type $[n, r, k, d]$ (where $d$ is not fixed) since the probability that a uniformly random code of such a type is dually QMDS tends to $ 1 $ as $ q $ tends to infinity.    
\end{corollary}

\begin{remark}
In particular, the probability that a uniformly random code of type $ [n,r,k,d] $ (where $ d $ is not fixed) is QMDS but not dually QMDS tends to $ 0 $ as $ q $ tends to infinity. That is, the family of such codes is sparse.
\end{remark}

\begin{remark}
As shown in \cite{gruica}, maximum rank distance (MRD) codes which are linear over a subfield are not dense in general. Together with Corollary \ref{cor density}, this shows an essential difference between MRD and MDS codes that are linear over subfields. The question remains open in general for the sum-rank distance \cite{fnt} (for MSRD codes that are linear over a subfield).
\end{remark}

We now turn to an explicit construction that works for any choice of $ n,r,k $, called polynomial ideal codes, introduced in \cite{bhandari}, and their generalization \cite{noga}. However, their field size $ q $ is far from optimal in general (we will explore small field sizes in Section \ref{sec long qmds}).

\begin{definition} [\textbf{Polynomial ideal codes} \cite{bhandari, noga}] \label{def pi codes}
Let $ F_1, F_2, \ldots, F_n \in \mathbb{F}_q[x] $ be pairwise coprime polynomials of degree $ r $. By the Chinese remainder theorem, we have the $ \mathbb{F}_q $-linear vector space isomorphism
$$ \varphi : \frac{\mathbb{F}_q[x]}{\left( \prod_{i=1}^n F_i \right)} \longrightarrow \mathbb{F}_q^{rn} : F \mapsto \left( F \textrm{ mod } F_i \right)_{i=1}^n, $$
where $ F $ mod $ G $ denotes the remainder of the Euclidean division of $ F $ by $ G $, which we identify with a vector in $ \mathbb{F}_q^{\deg(G)} $ by writing its coefficients as a list. For $ k \in [rn] $, we define the polynomial ideal code
$$ \mathcal{C}^{PI}_k(F_1, F_2, \ldots, F_n) = \{ \varphi(F) : F \in \mathbb{F}_q[x], \deg(F) < k \}. $$
Given also $ A_1, \ldots, A_n \in {\rm GL}_r(\mathbb{F}_q) $, we define the generalized polynomial ideal code
$$ \mathcal{C}^{PI}_k(F_1, \ldots, F_n; A_1, \ldots, A_n) = \{ (\mathbf{c}_1A_1, \ldots, \mathbf{c}_nA_n) : (\mathbf{c}_1, \ldots, \mathbf{c}_n) \in \mathcal{C}^{CR}_k(F_1, \ldots, F_n) \}. $$
\end{definition}

Polynomial ideal codes recover as particular cases codes which are known to have good list-decoding properties, such as folded Reed--Solomon codes and multiplicity codes, see \cite{bhandari}.

Clearly, generalized polynomial ideal codes are equivalent to polynomial ideal codes by Proposition \ref{prop isometries}. However, they will be necessary to express duals, as we show later.

First, we show that generalized polynomial ideal codes are always QMDS.

\begin{theorem} \label{th polynomial ideal qmds}
The generalized polynomial ideal code $ \mathcal{C}^{PI}_k(F_1, \ldots, F_n; A_1, \ldots, A_n) $ from Definition \ref{def pi codes} is QMDS of type $ [n,r,k,d] $, $ d = n - \lceil \frac{k}{r} \rceil +1 $.
\end{theorem}
\begin{proof}
It is enough to prove the result for polynomial ideal codes (i.e., $ A_i $ the identity matrix for all $ i \in [n] $), by Proposition \ref{prop isometries}.

First, clearly the code is linear. We now show that it has dimension $ k $. It is enough to show that if $ F \in \mathbb{F}_q[x] $ of $ \deg(F) < k $ satisfies $ \varphi(F) = 0 $, then $ F = 0 $, since the vector space of polynomials of degree less than $ k $ has dimension $ k $. The condition $ \varphi(F) = 0 $ means that $ \prod_{i=1}^n F_i $ divides $ F $ since $ F_1, F_2, \ldots, F_n $ are pairwise coprime. However, $ \deg(\prod_{i=1}^n F_i) = rn \geq k > \deg(F) $, thus it must hold that $ F = 0 $.

Now we show the QMDS property. Let $ F \in \mathbb{F}_q[x] $ be such that $ F \neq 0 $ and $ \deg(F) < k $. Assume that $ d = {\rm d}_F(\mathcal{C}^{PI}_k(F_1, F_2, \ldots, F_n)) = {\rm w}_F(\varphi(F)) $. Without loss of generality, we may assume that $ F $ mod $ F_i = 0 $, for $ i \in [n-d] $. That is, $ \prod_{i=1}^{n-d} F_i $ divides $ F $ since $ F_1, F_2, \ldots, F_{n-d} $ are pairwise coprime. Therefore
$$ r(n-d) = \deg \left( \prod_{i=1}^{n-d} F_i \right) \leq \deg(F) < k \leq r(n-d+1), $$
where the last inequality is the Singleton bound (Proposition \ref{singleton}). This means that $ \lceil \frac{k}{r} \rceil = n-d+1 $, i.e., $ \mathcal{C}^{PI}_k(F_1, F_2, \ldots, F_n) $ is QMDS and we are done.
\end{proof}

It remains to show that duals of generalized polynomial ideal codes are again QMDS. Such duals are hard to describe even in particular cases \cite{noga}. However, when $ F_i $ completely factorizes in $ \mathbb{F}_q $, then the duals are of the same form \cite[Th. 3.4]{noga}.

\begin{theorem}[\cite{noga}] \label{th polynomial ideal duals}
Let $ a_{i,j} \in \mathbb{F}_q $, for $ i \in [n] $ and $ j \in [r] $, be such that $ \{ a_{i,1}, \ldots, a_{i,r} \} \cap \{ a_{j,1}, \ldots, a_{j,r} \} = \varnothing $ if $ i \neq j $ and such that $ F_i = (x-a_{i,1}) \cdots (x-a_{i,r}) $, for $ i \in [n] $. Then for all $ A_1, \ldots, A_n \in {\rm GL}_r(\mathbb{F}_q) $, there exist $ B_1, \ldots, B_n \in {\rm GL}_r(\mathbb{F}_q) $ such that
$$ \mathcal{C}^{PI}_k(F_1, \ldots, F_n; A_1, \ldots, A_n)^\perp =  \mathcal{C}^{PI}_k(F_1, \ldots, F_n; B_1, \ldots, B_n) . $$
\end{theorem}

Combining Theorems \ref{th polynomial ideal qmds} and \ref{th polynomial ideal duals}, we deduce the following.

\begin{corollary} \label{cor PI are dually qmds}
Let $ a_{i,j} \in \mathbb{F}_q $, for $ i \in [n] $ and $ j \in [r] $, be such that $ \{ a_{i,1}, \ldots, a_{i,r} \} \cap \{ a_{j,1}, \ldots, a_{j,r} \} = \varnothing $ if $ i \neq j $ and such that $ F_i = (x-a_{i,1}) \cdots (x-a_{i,r}) $, for $ i \in [n] $. Then for all $ A_1, \ldots, A_n \in {\rm GL}_r(\mathbb{F}_q) $, the linear code $ \mathcal{C}^{PI}_k(F_1, \ldots, F_n; A_1, \ldots, A_n) $ is dually QMDS. 
\end{corollary}

Notice that these codes may be constructed for $ n = q $ (choosing $ a_{i,1} = \ldots = a_{i,r} $, for all $ i \in [n] $, which yields multiplicity codes). Observe that MDS codes (such that $ r | k $) exist for $ n = q^r + 1 $. In Section \ref{sec long qmds}, we will construct dually QMDS codes for $ n = 2^{r+1} - 1 $ over $ \mathbb{F}_2 $ for some dimensions that are not multiples of $ r $. 

\begin{remark}
We will see in Proposition \ref{prop qmds subcode is qmds} that a subcode of a QMDS code of the right dimension is again QMDS. In this way, one can obtain QMDS of any dimension for any length $ n \leq q^r + 1 $ by using extended Reed--Solomon codes. However, as we show in Remark \ref{remark dually qmds subcode not dually}, such subcodes are not always dually QMDS. To the best of our knowledge, the codes in Corollary \ref{cor PI are dually qmds} are the only dually QMDS codes that cover all possible lengths and dimensions (given sufficiently large fields).
\end{remark}

\section{Weight distributions} \label{sec weight dist}

In this section, we study weight distributions of linear codes in the folded Hamming distance.

\begin{definition}
The (folded) weight distribution of a linear code $ \mathcal{C} \subseteq \mathbb{F}_q^{rn} $ is defined as the numbers $A_j\left(\mathcal{C}\right)=\lvert\left\{\mathbf{c}\in\mathcal{C}: \text{w}_F(\mathbf{c})=j\right\}\rvert$, for $ j \in [0,n] $.
\end{definition}

We start by computing the weight distribution of dually QMDS codes. As in the classical MDS case, the weight distribution of dually QMDS codes only depends on their parameters, i.e., two dually QMDS codes of the same parameters have the same distribution. We later obtain a stronger result (Theorem \ref{th distribution degree freedom}) using MacWilliams equations. However, we give now a simple proof that only relies on Theorem \ref{th restricted dually qmds}.

\begin{theorem}
\label{dist}
Let $\C$ be a dually QMDS code of type $[n,r,k,d]$. Then $A_0=1$, $A_j=0$ for all $j\in[d-1]$, and if $ j \in [d,n] $, then
$$ A_j = \binom{n}{j} \sum_{i=0}^{j - d}(-1)^i \binom{j}{i} \big(q^{k-r(n-j+i)}-1\big). $$  
\end{theorem}
\begin{proof}
The only non-trivial cases are those with $ j \in [d,n] $. By Theorem \ref{th restricted dually qmds}, for $ I \subseteq [n] $,
\begin{equation}
|\C_I| = \left\lbrace \begin{array}{ll}
q^{k-r(n - |I|)} & \textrm{ if } r(n-|I|) \leq k, \\
1 & \textrm{ if } r(n-|I|) > k.
\end{array} \right.
\label{eq dist dually qmds}
\end{equation}
Note that $ N_t = \sum_{|I| = n-t} |\C_I| $ counts the number of words in $ \C $ with folded weight $ \leq n-t $, counted once for each $ \C_I $ they appear in, where $ \lvert I\rvert = n-t $. By (\ref{eq dist dually qmds}), we have
$$ N_t = \left\lbrace \begin{array}{ll}
\binom{n}{t} q^{k-rt} & \textrm{ if } rt\leq k, \\
\binom{n}{t} & \textrm{ if } rt> k.
\end{array} \right. $$
By the inclusion-exclusion principle, we deduce that
$$ A_j = \sum_{i=0}^{j} (-1)^i \binom{n-j+i}{i} N_{n-j+i} = $$
$$ \sum_{i=0}^{\lfloor\frac{k}{r}\rfloor+j-n}(-1)^i\binom{n-j+i}{i}\binom{n}{n-j+i}q^{k-r(n-j+i)}+\sum_{i=\lfloor\frac{k}{r}\rfloor+j-n+1}^{j}(-1)^i\binom{n-j+i}{i}\binom{n}{n-j+i} . $$
Using $ \binom{n-j+i}{i} \binom{n}{n-j+i} = \binom{n}{j}\binom{j}{i} $ and $\sum_{i=j-v+1}^{j}(-1)^i\binom{j}{i}=-\sum_{i=0}^{j-v}(-1)^i\binom{j}{i}$, we get
$$ A_j=\binom{n}{j} \sum_{i=0}^{\lfloor\frac{k}{r}\rfloor+j-n}(-1)^i\binom{j}{i} \big(q^{k-r(n-j+i)}-1\big), $$ 
for $ j \in [d,n] $. Now, if $ r \nmid k $, then we obtain the desired formula since $\lfloor\frac{k}{r}\rfloor+j-n=\lceil\frac{k}{r}\rceil-1+j-n=j-d$. Finally, if $ r \mid k $, then $\lfloor\frac{k}{r}\rfloor+j-n=\frac{k}{r}-1+j-n=j-d+1$, but the $ (j-d+1) $th term in the sum is zero, and we obtain again the desired formula.
\end{proof}

In fact, when the code is not dually QMDS, we have certain degrees of freedom for the weight distribution of the code. We show this stronger result in Theorem \ref{th distribution degree freedom} below. To prove it, we will need MacWilliams equations for the folded Hamming distance, which are of interest on their own and are new to the best of our knowledge.  

\begin{theorem}[\textbf{MacWilliams Equations}]\label{mw}
Let $\mathcal{C}$ be a code of type $[n,r,k,d]$, and let $A_j=A_j\left(\mathcal{C}\right)$ and $A_j^{\perp}=A_j\left(\mathcal{C}^{\perp}\right)$, for $ j \in [0,n] $. Then, for $ v \in [0,n] $,
$$\sum_{j=0}^{n-v}\binom{n-j}{v} A_j=q^{k-rv} \sum_{j=0}^{v}\binom{n-j}{n-v} A_j^{\perp}. $$
\end{theorem}
\begin{proof}
Denote $k={\rm dim}\left(\mathcal{C}\right)$ and $k_I={\rm dim}\left(\mathcal{C}^I\right)$, for $ I \subseteq [n] $. We have
    \begin{equation}
    \label{e1}
    \sum_{j=0}^{n-v}\binom{n-j}{v} A_j=\sum_{\lvert I\rvert=v}q^{k-k_I},
    \end{equation}
    for $ v \in [0,n] $, since both sides count the pairs $(\mathbf{c},I)$ such that $\mathbf{c}\in\mathcal{C}$, $I\subseteq [n]$, $\lvert I\rvert=v$, and $\pi_I(\mathbf{c})= \mathbf{0} $. First, for $\mathbf{c} \in \mathcal{C} $ such that $\w_F(\mathbf{c})=j$, there are $\binom{n-j}{v}$ possible sets $ I \subseteq [n] $ with $ |I| = v $ and $ \pi_I(\mathbf{c}) = \mathbf{0} $. Second, given $ I \subseteq [n] $ with $ |I| = v $, there are $q^{k-k_I}$ words in $\mathcal{C}$ with zeros in the $v$ blocks of $I$, since ${\rm dim}\left(\mathcal{C}\cap\ker (\pi_I) \right)= {\rm dim}\left(\mathcal{C}\right)-{\rm dim}\left(\pi_I(\mathcal{C})\right) = k-k_I$. 
    
    Next, the right-hand side of (\ref{e1}) equals
    \begin{equation}
    \label{e2}
    q^{k-rv}\sum_{ \lvert I\rvert=v} \left\lvert \left(\mathcal{C}^I\right)^{\perp} \right\rvert = q^{k-rv}\sum_{ \lvert I\rvert=v}\sum_{j=0}^v A_j\left(\left(\mathcal{C}^I\right)^{\perp}\right)= q^{k-rv}\sum_{j=0}^v\sum_{ \lvert I\rvert=v} A_j\left(\left(\mathcal{C}^I\right)^{\perp}\right).
    \end{equation}
    
    Now we prove that
    \begin{equation}
    \label{e3}\sum_{\lvert I\rvert=v}A_j\left(\left(\mathcal{C}^{\perp}\right)_I\right)=\binom{n-j}{v-j} A_j^{\perp}.
    \end{equation}
    Both numbers count the possible pairs $(\mathbf{c},I)$ such that $\mathbf{c}\in\mathcal{C}^{\perp}$, $\text{w}_F(\mathbf{c})= j$, $I\subseteq [n]$, $\lvert I\rvert=v$, and $\pi_{[n]\setminus I}(\mathbf{c}) = \mathbf{0} $. First, for $ I \subseteq [n] $ with $ |I| = v $, there are $ A_j\left(\left(\mathcal{C}^{\perp}\right)_I\right)$ words in $\mathcal{C}^{\perp}$ with weight $j$ and with zeros outside the $v$ blocks of $I$. Second, for $\mathbf{c}\in\mathcal{C}^{\perp}$ such that $\w_F(\mathbf{c})=j$, there are $\binom{n-j}{n-v} = \binom{n-j}{v-j}$ possible sets $ I \subseteq [n] $ such that $ |I| = v $ and $\pi_{[n]\setminus I}(\mathbf{c}) = \mathbf{0} $. 
    
The theorem follows by using that (\ref{e1}) and (\ref{e2}) are equal, and applying (\ref{e3}) together with $\left(\mathcal{C}^I\right)^{\perp}=\left(\mathcal{C}^{\perp}\right)_I$ (see \cite[Th. 1.5.7]{pless}).
\end{proof}

In order to give the strengthening of Theorem \ref{dist}, we need the following preliminary lemma.

\begin{lemma} \label{lemma M N inverses}
Consider the matrices $ M_n = \left( (-1)^{n-i+j} \binom{j}{n-i} \right)_{i=0,j=0}^{n,n} $ and $ N_n = \left( \binom{n-j}{i} \right)_{i=0,j=0}^{n,n} $, of size $ (n+1) \times (n+1) $ over $ \mathbb{Z} $, where $ \binom{u}{v} = 0 $ if $ v > u $. Then $ N_n = M_n^{-1} $.
\end{lemma}
\begin{proof}
We prove that $ M_n N_n = I_{n+1} $, the identity of size $ n+1 $. The product of the $ i $th row of $ M_n $ and the $ j $th column of $ N_n $ is
\begin{equation}
\sum_{\ell = 0}^n (-1)^{n-i+\ell} \binom{\ell}{n-i} \binom{n-j}{\ell}.
\label{eq product i row N j col M}
\end{equation}
If $ i < j $, then clearly (\ref{eq product i row N j col M}) equals $ 0 $. If $ i=j $, then (\ref{eq product i row N j col M}) becomes $ (-1)^{2(n-i)}\binom{n-i}{n-i}\binom{n-i}{n-i} = 1 $. Finally, if $ i > j $, then (\ref{eq product i row N j col M}) becomes 
\begin{equation}
\sum_{\ell = n-i}^{n-j} (-1)^{n-i+\ell} \binom{\ell}{n-i} \binom{n-j}{\ell} = \sum_{\ell = 0}^{i-j} (-1)^\ell \binom{n-i+\ell}{n-i} \binom{n-j}{i-j-\ell}.
\label{eq product i row N j col M 2}
\end{equation}
Now, we have that $ (1+x)^{-n+i-1} = \sum_{u=0}^\infty (-1)^u \binom{n-i+u}{n-i} x^u $ (see \cite[p. 56]{aigner}) and $ (1+x)^{n-j} = \sum_{v=0}^{n-j} \binom{n-j}{v} x^v $. Using these series expansions in $ \mathbb{Z}[\![ x ]\!] $, it is straightforward to see that the right-hand side of (\ref{eq product i row N j col M 2}) coincides with the $ (i-j) $th coefficient of $ (1+x)^{-n+i-1} (1+x)^{n-j} = (1+x)^{i-j-1} $, which is zero. Therefore (\ref{eq product i row N j col M 2}) equals zero if $ i > j $ and we are done.
\end{proof}

We may now prove the above mentioned strengthening of Theorem \ref{dist}.

\begin{theorem} \label{th distribution degree freedom}
Let $ \mathcal{C} $ be a linear code of type $ [n,r,k,d] $, and denote $ d^\perp = \dd(\mathcal{C}) $ and $ A_j = A_j(\mathcal{C}) $, for $ j \in [0,n] $. Then for $ j \in [n-d^\perp +1,n] $, it holds
$$ A_j = \sum_{i=0}^{j- n + d^\perp - 1} (-1)^i \left( \binom{n}{j} \binom{j}{i} (q^{k-r(n-j+i)} - 1) - \sum_{v=d}^{n-d^\perp} \binom{n-j+i}{i} \binom{n-v}{n-j+i} A_v \right). $$
In particular, the whole weight distribution of $ \mathcal{C} $ is determined by $ A_d, A_{d+1}, \ldots, A_{n-d^\perp} $.
\end{theorem}
\begin{proof}
Denote $ A^\perp_j = A_j(\mathcal{C}^\perp) $. Using $ A_0 = A^\perp_0 = 1 $, $ A_1 = \ldots = A_{d-1} = 0 $ and $ A^\perp_1 = \ldots = A^\perp_{d^\perp-1} = 0 $ in Theorem \ref{mw}, we have, for $ u \in [0,d^\perp -1] $,
$$ \sum_{j=n-d^\perp+1}^{n-u} \binom{n-j}{u} A_j = \binom{n}{u} (q^{k-ru}-1) - \sum_{v = d}^{n-d^\perp} \binom{n-v}{u} A_v. $$
We can rewrite these equations in matrix form as
$$ \left( \begin{array}{ccccc}
\binom{d^\perp-1}{0} & \binom{d^\perp-2}{0} & \ldots & \binom{1}{0} & \binom{0}{0} \\
\binom{d^\perp-1}{1} & \binom{d^\perp-2}{1} & \ldots & \binom{1}{1} & 0 \\
\vdots & \vdots & \ddots & \vdots & \vdots \\
\binom{d^\perp-1}{d^\perp -1} & 0 & \ldots & 0 & 0 
\end{array} \right) \left( \begin{array}{c}
A_{n-d^\perp+1} \\
A_{n-d^\perp+2} \\
\vdots \\
A_n
\end{array} \right) $$
$$ = \left( \begin{array}{ccccc}
\binom{n}{0} (q^k - 1) & \binom{n-d}{0} & \binom{n-d+1}{0} & \ldots & \binom{d^\perp}{0} \\
\binom{n}{1} (q^{k-r}-1) & \binom{n-d}{1} & \binom{n-d+1}{1} & \ldots & \binom{d^\perp}{1} \\
\vdots & \vdots & \vdots & \ddots & \vdots \\
\binom{n}{d^\perp -1} (q^{k - r(d^\perp-1)}-1) & \binom{n-d}{d^\perp-1} & \binom{n-d+1}{d^\perp-1} & \ldots & \binom{d^\perp}{d^\perp-1}  
\end{array} \right) \left( \begin{array}{c}
1 \\
- A_d \\
\vdots \\
- A_{n-d^\perp}
\end{array} \right). $$
Since the matrix on the left-hand side is $ N_{d^\perp-1} $, we may solve such linear equations by multiplying by $ M_{d^\perp-1} $ on the left on both sides, by Lemma \ref{lemma M N inverses}. This tedious calculation yields
\begin{equation*}
\begin{split}
A_j & = \sum_{i=n-j}^{d^\perp - 1} (-1)^{n-j+i} \left( \binom{i}{n-j} \binom{n}{i} (q^{k-ir}-1) - \sum_{v=d}^{n-d^\perp} \binom{i}{n-j} \binom{n-v}{i} A_v \right) \\
& = \sum_{i=0}^{j- n + d^\perp - 1} (-1)^i \left( \binom{n-j+i}{n-j} \binom{n}{n-j+i} (q^{k-(i+n-j)r}-1) - \sum_{v=d}^{n-d^\perp} \binom{n-j+i}{n-j} \binom{n-v}{n-j+i} A_v \right),
\end{split}
\end{equation*}
for $ j \in [n-d^\perp +1,n] $, and the theorem follows by using $ \binom{n-j+i}{n-j} \binom{n}{n-j+i} = \binom{n}{j}\binom{j}{i} $.
\end{proof}

\begin{remark}
In the case of dually QMDS codes, we have $ d+d^\perp \in \{n+1,n+2 \} $ by Proposition \ref{prop d+d dual}. In that case, we have $ d > n-d^\perp $ and clearly Theorem \ref{th distribution degree freedom} recovers Theorem \ref{dist} (notice that for linear MDS codes, Theorem \ref{th distribution degree freedom} gives $ A_{n-d^\perp+1} = A_{d-1} = 0 $ as expected). In all other cases, it holds that $ d+d^\perp \leq n $ and we have $ n - d - d^\perp + 1 \geq 1 $ degrees of freedom for the weight distribution of the code according to Theorem \ref{th distribution degree freedom}.
\end{remark}

We conclude by noting that QMDS codes that are not dually QMDS never have the same weight distribution as dually QMDS codes. This fact follows by combining Theorems \ref{dist} and \ref{mw}. It also follows by using Theorem \ref{th restricted dually qmds} and noting that the first proof we gave of Theorem \ref{dist} only depends on the dimensions of the restricted codes.

\begin{corollary}
A linear code $ \mathcal{C} $ of type $ [n,k,r,d] $ is dually QMDS if, and only if, its weight distribution $ A_j(\mathcal{C}) $, $ j \in [n] $, is given as in Theorem \ref{dist}.
\end{corollary}

\begin{example}
Consider the code $ \mathcal{C} $ of type $ [3,3,4,2] $ from Example \ref{example qmds 1} over $ \mathbb{F}_2 $, which is QMDS but not dually QMDS. We observe that $ A_2(\mathcal{C}) \geq 4 $, since it has the following codewords of folded weight $ 2 $: $ (1,0,0|1,0,0|0,0,0) $, $ (0,1,0|0,1,0|0,0,0) $, $ (0,0,1|0,0,0|0,1,0) $ and $ (1,1,0|0,0,0|1,0,0) $. However, according to Theorem \ref{dist}, a dually QMDS code of type $ [3,3,4,2] $ over $ \mathbb{F}_2 $ satisfies $ A_2 = 3 $. Since there exists a dually QMDS code of type $ [3,3,4,2] $ over $ \mathbb{F}_2 $ (a restriction of the code in Construction \ref{const long qmds 2^(r+1)-1} for $ r = 3 $), we conclude that there exist two QMDS codes with the same parameters and different weight distributions, in contrast with dually QMDS codes.
\end{example}

\section{Long QMDS codes over small fields} \label{sec long qmds}

In this section, we study how long QMDS and dually QMDS codes can be over a given (preferably small) finite field. We give two bounds and an optimal-length binary construction.

We start with a simple upper bound on the length $ n $ of QMDS codes given by the distance $ d $ and the field size. This result is inspired by the discussion at the beginning of \cite{singleton} and we prove it for arbitrary (linear or nonlinear) codes.

\begin{theorem} \label{th bound on n by d}
Let $ \mathcal{C} \subseteq \mathbb{F}_q^{rn} $ be a (linear or nonlinear) code of size $ q^k $ and minimum distance $ d = n - \lceil \frac{k}{r} \rceil + 1 \geq 3 $ (in particular, if $ \mathcal{C} $ is linear, then it is QMDS). Then
$$ n \leq d - 3 + q^{r \lceil \frac{k}{r} \rceil - k} (q^r + 1) \leq d - 3 + q^{r-1}(q^r+1). $$
In particular, $ \lceil \frac{k}{r} \rceil \leq q^{r \lceil \frac{k}{r} \rceil - k} (q^r + 1) - 2 \leq q^{r-1}(q^r+1) -2 $. 
\end{theorem}
\begin{proof}
Let $ I = [n] \setminus [d-3] $. Clearly $ |\mathcal{C}^I| = q^k $ and $ \dd (\mathcal{C}^I) \geq 3 $. Thus we may apply the classical Hamming bound to $ \mathcal{C}^I $, i.e., the balls $ \mathcal{B}(\mathbf{c},1) $ of folded Hamming radius $ 1 $ around every codeword $ \mathbf{c} \in \mathcal{C}^I $ are pairwise disjoint. Thus
$$ q^k \left( 1 + (n-d+3)(q^r-1) \right) = |\mathcal{C}^I| \cdot |\mathcal{B}(\mathbf{0},1)| \leq | \mathbb{F}_q^{r|I|}| = q^{r(n-d+3)}. $$
This inequality can be rearranged as follows, 
\begin{equation*}
\begin{split}
n & \leq d - 3 + \frac{q^{r(n-d+3) - k} - 1}{q^r - 1} \\
 & = d - 3 + \frac{q^{r(\lceil \frac{k}{r} \rceil + 2) - k} - 1}{q^r - 1} \\
 & = d - 3 + q^{r \lceil \frac{k}{r} \rceil - k} \cdot \frac{q^{2r} - 1}{q^r - 1} + \frac{q^{r \lceil \frac{k}{r} \rceil - k} - 1}{q^r-1} \\
 & < d - 3 + q^{r \lceil \frac{k}{r} \rceil - k} (q^r + 1) + 1,
\end{split}
\end{equation*}
where in the last inequality we use that $ 0 \leq r \lceil \frac{k}{r} \rceil - k < r $, and we are done.
\end{proof}

Notice that if $ d $ is unrestricted, then the previous bound does not imply that $ n $ is restricted by $ q $ or $ r $ (although $ k $ is). However, in the case of dually QMDS codes, we can obtain bounds on $ d $ and $ n $ in terms of $ q $ and $ r $ based on their weight distributions (Theorem \ref{dist}). This result generalizes \cite[Cor. 7.4.3]{pless} from $ r = 1 $ to $ r \geq 1 $.  

\begin{theorem} 
\label{439} Let $\C$ be a dually QMDS code of type $[n,r,k,d]$ and denote $ d^\perp = d(\C^{\perp}) $. Let $ \varepsilon = r - (r \lceil \frac{k}{r} \rceil - k) \in [r] $ and $ \delta = r - (k - r \lfloor \frac{k}{r} \rfloor) \in [r] $.
\begin{enumerate}
    \item If $ k > r $, then $ d \leq q^r-1+ \left\lfloor \frac{q^r-1}{q^\varepsilon -1} \right\rfloor $.
    \item If $ k < r(n-1) $, then $ d^\perp = \lceil \frac{k+1}{r} \rceil \leq q^r-1+ \left\lfloor \frac{q^r-1}{q^\delta -1} \right\rfloor $. 
\end{enumerate}
In particular, if $ r < k < r(n-1) $, then 
$$ n \leq \left\lbrace \begin{array}{ll}
2q^r - 2 & \textrm{ if } r \mid k, \\
2 q^r - 3 + \left\lfloor \frac{q^r-1}{q^\varepsilon-1} \right\rfloor + \left\lfloor \frac{q^r-1}{q^\delta - 1} \right\rfloor & \textrm{ if } r \nmid k.
\end{array} \right. $$
\end{theorem}
\begin{proof}
The bound on $ n $ follows by adding the bounds in Items 1 and 2 and using Proposition \ref{prop d+d dual} for $ d+d^\perp $. Moreover, Item 2 is the dual statement of Item 1, hence we only prove the latter.

If $ r \nmid k $, then $ k > r $ implies $ d^\perp \geq 2 $ and $ d = n+1-d^\perp \leq n-1 $ by Proposition \ref{prop d+d dual}. If $ r \mid k $, then $ k > r $ implies $ k \geq 2r $, hence $ d^\perp \geq 3 $, thus $ d = n+2-d^\perp \leq n-1 $ too. Therefore in both cases ($ r \mid k $ or not), we may consider $ A_{d+1} $. By Theorem \ref{dist}, we have that
$$ A_{d+1} = \binom{n}{d+1}\left(q^{k+r(d+1-n)}-1-(d+1)\big(q^{k+r(d-n)}-1\big)\right) \geq 0. $$
Now, we have $ k + r(d-n) = k - r (\lceil \frac{k}{r} \rceil -1) = \varepsilon $, and the previous inequality is equivalent to $ q^{\varepsilon + r} - 1 \geq (d+1) \left( q^{\varepsilon} - 1 \right) $. Therefore 
$$ d+1 \leq \frac{q^{\varepsilon + r} - 1}{q^{\varepsilon} - 1} = \frac{q^{\varepsilon +r} - q^r + q^r - 1}{q^{\varepsilon} - 1} = q^r + \frac{q^r-1}{q^{\varepsilon} - 1}. $$
\end{proof}

We now turn to constructions. The first observation is that QMDS codes of lower dimension may be easily obtained from a given QMDS code (e.g., a given linear MDS code), as follows.

\begin{proposition} \label{prop qmds subcode is qmds}
Let $ \mathcal{C} $ be a QMDS code of type $ [n,r,k,d] $ and let $ k^\prime $ be an integer such that $ r (\lceil \frac{k}{r} \rceil - 1) < k^\prime < k $ (thus $ \lceil \frac{k^\prime}{r} \rceil = \lceil \frac{k}{r} \rceil $). Then any linear subcode of $ \mathcal{C} $ of dimension $ k^\prime $ is QMDS of type $ [n,r,k^\prime,d] $. 
\end{proposition}

\begin{corollary} \label{cor unrestricted qmds codes d=2}
Let $ \mathcal{C} $ be a linear MDS code of type $ [n,r,r(n-1),2] $, which exists over any field (take, e.g., the dual of the classical repetition code). For any integer $ k $ with $ r(n-2) < k < r(n-1) $, any linear subcode of $ \mathcal{C} $ of dimension $ k $ is a QMDS code of type $ [n,r,k,2] $. 
\end{corollary}

Thus there exist QMDS codes of dimension $ r(n-2) < k < r(n-1) $ for any length $ n $ over any field. Thus the last bound in Theorem \ref{439} does not hold for general QMDS codes when $ r(n-2) < k < r(n-1) $. Moreover, since the distance is $ 2 $, this result also shows that the bound on $ n $ in Theorem \ref{th bound on n by d} does not hold if $ d \geq 3 $ is not satisfied.

\begin{remark} \label{remark dually qmds subcode not dually}
Notice that the QMDS codes from Corollary \ref{cor unrestricted qmds codes d=2} may not be dually QMDS codes if $ n $ is longer than the bound in Theorem \ref{439}. In particular, one may not always obtain a dually QMDS code by choosing a subcode of a dually QMDS code as in Proposition \ref{prop qmds subcode is qmds}.
\end{remark}

Next we give a construction of optimal-length binary dually QMDS codes, longer than any MDS code. We first give an example.

\begin{example}
The linear code over $ \mathbb{F}_2 $ with the following generator matrix is dually QMDS of type $ [7,2,3,6] $:
$$ G = \left( \begin{array}{cc|cc|cc|cc|cc|cc|cc}
1 & 0 & 1 & 0 & 1 & 0 & 1 & 0 & 1 & 0 & 1 & 0 & 0 & 0 \\
0 & 1 & 0 & 1 & 0 & 1 & 0 & 1 & 1 & 0 & 0 & 0 & 1 & 0 \\
1 & 1 & 1 & 0 & 0 & 1 & 0 & 0 & 0 & 1 & 0 & 1 & 0 & 1
\end{array} \right). $$
The distance of the code is $ 6 $ since any nontrivial $ \mathbb{F}_2 $-linear combination of the rows of $ G $ has exactly one zero block. Thus the code is QMDS. Moreover, since the rows of any block of $ G $ span the whole space $ \mathbb{F}_2^2 $, then the dual has distance $ 2 $ and thus the code is dually QMDS.
\end{example}

In order to provide the general construction, we need the following technical lemma.

\begin{lemma} \label{lemma for construction}
Let $ r $ be a positive integer and let $ \mathbf{e}_1, \ldots, \mathbf{e}_r \in \mathbb{F}_2^r $ be the vectors of the standard basis (i.e., $ e_{i,j} = \delta_{i,j} $). Let $ I = \{ i_1, \ldots, i_t \} \subseteq [r+1] $ with $ 1 \leq i_1 < i_2 < \ldots < i_t \leq r+1 $ and $ t \geq 1 $ (i.e., $ I \neq \varnothing $). Finally, define $ \mathbf{u}_{I,i} = \mathbf{e}_i $ for $ i \in [i_t-1] $, $ \mathbf{u}_{I,i} = \mathbf{e}_{i-1} $ for $ i \in [i_t+1, r+1] $, $ \mathbf{u}_{I,i_t} = \sum_{j=1}^{t-1} \mathbf{e}_{i_j} $ if $ t \geq 2 $ and $ \mathbf{u}_{I,i_t} = \mathbf{0} $ if $ t = 1 $. Then, for any nonempty $ J \subseteq [r+1] $, it holds that $ \sum_{i \in J} \mathbf{u}_{I,i} = \mathbf{0} $ if, and only if, $ J = I $.
\end{lemma}
\begin{proof}
First $ \sum_{i \in J} \mathbf{u}_{I,i} = 2 \sum_{j=1}^{t-1} \mathbf{e}_{i_j} = \mathbf{0} $ if $ J = I $. Now assume that $ J \neq I $. If $ i_t \notin J $, then all the components of $ \sum_{i \in J} \mathbf{u}_{I,i} $ corresponding to $ J \neq \varnothing $ equal $ 1 $, hence $ \sum_{i \in J} \mathbf{u}_{I,i} \neq \mathbf{0} $. If $ i_t \in J $, then there exists $ j \in I \setminus J $ or $ j \in J \setminus I $ with $ j \neq i_t $. If $ j < i_t $, then the $ j $th component of $ \sum_{i \in J} \mathbf{u}_{I,i} $ must be $ 1 $, and if $ j > i_t $, then the $ (j-1) $th component of $ \sum_{i \in J} \mathbf{u}_{I,i} $ must be $ 1 $, thus $ \sum_{i \in J} \mathbf{u}_{I,i} \neq \mathbf{0} $ and we are done.
\end{proof}

\begin{construction} \label{const long qmds 2^(r+1)-1}
Let $ r $ be a positive integer and enumerate all the nonempty subsets of $ [r+1] $ as $ I_1, I_2, \ldots, I_n $, where $ n = 2^{r+1}-1 $. Let $ \mathcal{C} \subseteq \mathbb{F}_2^n $ be the linear code with generator matrix
$$ G = \left( \begin{array}{c|c|c|c}
\mathbf{u}_{I_1,1} & \mathbf{u}_{I_2,1} & \ldots & \mathbf{u}_{I_n,1} \\
\vdots & \vdots & \ddots & \vdots \\
\mathbf{u}_{I_1,r+1} & \mathbf{u}_{I_2,r+1} & \ldots & \mathbf{u}_{I_n,r+1}
\end{array} \right), $$
where $ \mathbf{u}_{I_i,j} $ is as in the previous lemma, for $ i \in [n] $ and $ j \in [r+1] $.
\end{construction}

\begin{theorem} \label{th construction long QMDS}
The linear code in Construction \ref{const long qmds 2^(r+1)-1} is dually QMDS of type $ [2^{r+1}-1,r,r+1,2^{r+1}-2] $. 
\end{theorem}
\begin{proof}
Consider a nonzero codeword $ \mathbf{c} = (\mathbf{c}_1, \ldots, \mathbf{c}_n) \in \mathcal{C} $, where $ \mathbf{c}_i \in \mathbb{F}_2^r $, for $ i \in [n] $. There exists a nonzero $ \mathbf{x} = (x_1, \ldots, x_{r+1}) \in \mathbb{F}_2^{r+1} $ such that $ \mathbf{c} = \mathbf{x}G $. Define $ J = \{ j \in [r+1] : x_j \neq 0 \} $. Since $ J \neq \varnothing $, there exists $ i \in [n] $ such that $ J = I_i $. By Lemma \ref{lemma for construction}, we deduce that $ \mathbf{c}_i = \sum_{j \in J} \mathbf{u}_{J,j} = \mathbf{0} $, whereas if $ \ell \in [n] \setminus \{ i \} $, then $ \mathbf{c}_\ell = \sum_{j \in J} \mathbf{u}_{I_\ell,j} \neq \mathbf{0} $ since $ J \neq I_\ell $. In other words, $ \w_F(\mathbf{c}) = n-1 $, and therefore $ \dd(\mathcal{C}) = n-1 $. Since the generator matrix $ G $ is clearly of full rank $ r+1 $, we conclude that $ \mathcal{C} $ is QMDS of type $ [2^{r+1}-1,r,r+1,2^{r+1}-2] $. Finally, for every $ i \in [n] $, the rows in the $ i $th block of $ G $, $ \mathbf{u}_{I_i,1} , \ldots, \mathbf{u}_{I_i,r+1} $, span the whole space $ \mathbb{F}_2^r $, thus there cannot be a codeword in $ \mathcal{C}^\perp $ of folded weight $ 1 $. We conclude that $ \dd(\mathcal{C}) \geq 2 $, which means that $ \mathcal{C}^\perp $ is also QMDS and $ \mathcal{C} $ is dually QMDS.
\end{proof}

\begin{remark}
The distance $ d = 2^{r+1} - 2 $ of the code in Construction \ref{const long qmds 2^(r+1)-1} attains the bound in Item 1 in Theorem \ref{439}. Observe that in that theorem, $ \varepsilon = 1 $, thus such a bound is
$$ q^r -1 + \frac{q^r-1}{q^{\varepsilon}-1} = 2^r - 1 + \frac{2^r-1}{2-1} = 2^{r+1}-2 = d. $$
Moreover, for an arbitrary MDS code over $ \mathbb{F}_q $ of distance $ d $ and size $ q^{rk} $, it is known \cite{singleton} that $ n \leq 2q^r-2 $ if $ r < k < r(n-1) $. In the case $ q=2 $, we obtain $ n \leq 2^{r+1}-2 $, but the code in Construction \ref{const long qmds 2^(r+1)-1} attains the length $ n = 2^{r+1}-1 $, thus is longer than any MDS code. Tighter bounds for MDS codes hold for prime fields \cite{ball}, coinciding with the MDS conjecture. 
\end{remark}

From the proof of Theorem \ref{th construction long QMDS} we also conclude the following property. 

\begin{corollary}
The linear code in Construction \ref{const long qmds 2^(r+1)-1} is a one-weight or constant-weight code, that is, all of its nonzero codewords have the same weight.
\end{corollary}

Furthermore, these are the longest QMDS codes of dimension $ k $ with $ r+1 \leq k \leq 2r $ over $ \mathbb{F}_2 $.

\begin{proposition} \label{prop longest qmds F_2 r+1}
Let $ r $ and $ k $ be positive integers such that $ r+1 \leq k \leq 2r $. If there exists a QMDS code of type $ [n,r,k,n-1] $ over $ \mathbb{F}_2 $, then $ n \leq 2^{r+1}-1 $.
\end{proposition}
\begin{proof}
By choosing an $ (r+1) $-dimensional subcode, we may assume that there exists a QMDS code of type $ [n,r,r+1,n-1] $ over $ \mathbb{F}_2 $ by Proposition \ref{prop qmds subcode is qmds}. Consider a generator matrix of the code $ G = (G_1 | \ldots |G_n) $, where $ G_i $ is a binary $ (r+1) \times r $ matrix. If $ \mathbf{v}_{i,1} , \ldots, \mathbf{v}_{i,r+1} \in \mathbb{F}_2^r $ denote the rows of $ G_i $, then there must exist a nonempty $ I_i \subseteq [r+1] $ such that $ \sum_{j \in I_i} \mathbf{v}_{i,j} = \mathbf{0} $. Therefore, given a nonzero $ \mathbf{x} = (x_1, \ldots, x_{r+1}) \in \mathbb{F}_2^{r+1} $, the codeword $ \mathbf{c} = \mathbf{x}G $ is zero at least in the $ i $th block if it holds that $ I_i = \{ j \in [r+1] : x_j = 0 \} $ ($ I_i $ is the support of $ \mathbf{x} $). Since the code has distance $ n - 1 $, then the nonempty sets $ I_1, \ldots , I_n $ must be all distinct, and thus $ n \leq 2^{r+1}-1 $.
\end{proof}

In fact, for $ r+2 \leq k \leq 2r $, lengths $ n = 2^{r+1} -1 $ are not achievable over $ \mathbb{F}_2 $ by dually QMDS codes.

\begin{proposition}
Assume that $ r \geq 2 $ and $ r+2 \leq k \leq 2r $. If there exists a dually QMDS code of type $ [n,r,k,n-1] $ over $ \mathbb{F}_2 $, then 
$$ n \leq \frac{4}{3} (2^r-1) + 1 < 2^{r+1} - 2 . $$
\end{proposition}
\begin{proof}
By Theorem \ref{439}, we have that
$$ n-1 = d \leq 2^r - 1 + \frac{2^r-1}{2^\varepsilon -1}, $$
where $ \varepsilon = r - (r \lceil \frac{k}{r} \rceil - k ) = r - (2r - k) = k-r \geq 2 $. Therefore, we conclude that
$$ n-1 \leq 2^r - 1 + \frac{2^r-1}{4-1} = \frac{4}{3} (2^r-1), $$
and we are done.
%
%
\end{proof}

We conclude the section with some examples of binary linear QMDS codes that are longer than any possible MDS code.

\begin{example}
The linear code over $ \mathbb{F}_2 $ with the following generator matrix is QMDS of type $ [9,2,13,3] $:
$$ G = \left( \begin{array}{cccccccccccc|cc|cc|cc}
1 & 0 & 0 & 0 & 0 & 0 & 0 & 0 & 0 & 0 & 0 & 0 & 1 & 0 & 1 & 0 & 0 & 0 \\
0 & 1 & 0 & 0 & 0 & 0 & 0 & 0 & 0 & 0 & 0 & 0 & 0 & 1 & 0 & 0 & 0 & 1 \\
0 & 0 & 1 & 0 & 0 & 0 & 0 & 0 & 0 & 0 & 0 & 0 & 0 & 0 & 1 & 0 & 0 & 1 \\
0 & 0 & 0 & 1 & 0 & 0 & 0 & 0 & 0 & 0 & 0 & 0 & 1 & 0 & 1 & 1 & 0 & 0 \\
0 & 0 & 0 & 0 & 1 & 0 & 0 & 0 & 0 & 0 & 0 & 0 & 0 & 1 & 1 & 0 & 0 & 0 \\
0 & 0 & 0 & 0 & 0 & 1 & 0 & 0 & 0 & 0 & 0 & 0 & 0 & 0 & 0 & 1 & 0 & 1 \\
0 & 0 & 0 & 0 & 0 & 0 & 1 & 0 & 0 & 0 & 0 & 0 & 1 & 0 & 1 & 0 & 1 & 0 \\
0 & 0 & 0 & 0 & 0 & 0 & 0 & 1 & 0 & 0 & 0 & 0 & 1 & 1 & 0 & 0 & 0 & 1 \\
0 & 0 & 0 & 0 & 0 & 0 & 0 & 0 & 1 & 0 & 0 & 0 & 0 & 1 & 1 & 0 & 0 & 1 \\
0 & 0 & 0 & 0 & 0 & 0 & 0 & 0 & 0 & 1 & 0 & 0 & 0 & 0 & 0 & 1 & 1 & 1 \\
0 & 0 & 0 & 0 & 0 & 0 & 0 & 0 & 0 & 0 & 1 & 0 & 1 & 0 & 0 & 0 & 1 & 1 \\
0 & 0 & 0 & 0 & 0 & 0 & 0 & 0 & 0 & 0 & 0 & 1 & 1 & 1 & 0 & 1 & 0 & 1 \\
0 & 0 & 0 & 0 & 0 & 0 & 0 & 0 & 0 & 0 & 0 & 0 & 1 & 1 & 1 & 1 & 1 & 1 
\end{array} \right). $$
Notice that a (linear or nonlinear) MDS code of distance $ d = 3 $ must satisfy $ n \leq q^r+1 $ by Theorem \ref{th bound on n by d}. In the case $ q=r=2 $ as in this example, it must satisfy $ n \leq 5 $. However, the QMDS code of distance $ d=3 $ in this example satisfies $ n = 9 $. Note also that, for this code, $ \lceil \frac{k+1}{r} \rceil = 7 > 6 \geq 2^r - 1 + \lfloor \frac{2^r-1}{2^\delta -1} \rfloor $, thus the bound on $ k $ in Theorem \ref{439}, Item 2, does not hold for general QMDS codes. 
\end{example}

\begin{example}
The linear code over $ \mathbb{F}_2 $ with the following generator matrix is dually QMDS of type $ [6,2,5,4] $:
$$ G = \left( \begin{array}{cc|cc|cc|cc|cc|cc}
1 & 0 & 0 & 0 & 0 & 1 & 0 & 0 & 0 & 1 & 1 & 1 \\
0 & 1 & 0 & 0 & 0 & 0 & 1 & 1 & 0 & 1 & 0 & 1 \\
0 & 0 & 1 & 0 & 1 & 1 & 0 & 0 & 0 & 1 & 0 & 1 \\
0 & 0 & 0 & 1 & 0 & 0 & 1 & 1 & 1 & 0 & 1 & 0 \\
0 & 0 & 0 & 0 & 1 & 0 & 0 & 1 & 1 & 0 & 0 & 1 
\end{array} \right). $$
According to the MDS conjecture, a (linear or nonlinear) MDS code in $ \mathbb{F}_2^{rn} $ with $ r = 2 $ must satisfy $ n \leq 2^r+1 = 5 $. However, the dually QMDS code in this examples satisfies $ n = 6 $.
\end{example}

\section{Pseudo arcs: A geometric description} \label{sec pseudo arcs}

In this section, we provide a finite-geometry counterpart of linear codes in the folded Hamming distance, which coincides with what is called arcs or pseudo arcs in the finite-geometry literature \cite{ball-additive}. They generalize projective systems associated to linear codes in the classical Hamming distance \cite{tsfasman} and partial spreads \cite{beutel}. Furthermore, pseudo arcs are the building blocks of recent general families of MSRD and PMDS codes with small field sizes \cite{liu, generalMSRD}.

\begin{definition} \label{def pseudo arcs}
A pseudo arc of type $ [n,r,m,t] $ is a tuple $ \mathcal{H} = ( \mathcal{H}_i )_{i=1}^n $ such that
\begin{enumerate}
\item
$ \mathcal{H}_i \subseteq \mathbb{F}_q^m $ is an ($ \mathbb{F}_q $-linear) subspace of $ \dim(\mathcal{H}_i) = r $, for all $ i \in [n] $.
\item
$ t $ is the maximum positive integer such that $ \mathcal{H}_i \cap (\sum_{j \in J} \mathcal{H}_j ) = 0 $, for all $ i \in [n] $ and all $ J \subseteq [n] \setminus \{ i \} $ with $ |J| = t-1 $ (i.e., any $ t $ of the subspaces are in direct sum).
\end{enumerate}
\end{definition}

\begin{remark}
A partial spread \cite{beutel} is a pseudo arc of type $ [n,r,m, t] $ with $ t \geq 2 $.
\end{remark}

\begin{remark}
Definition \ref{def pseudo arcs} coincides with the definition at the beginning of \cite[Sec. 4]{ball-additive} after projectivization, except the parameter $ t $ is not considered to be maximum in \cite{ball-additive}.
\end{remark}

We need to consider nondegenerate pseudo arcs in order to associate them to linear codes in the folded Hamming distance.

\begin{definition} \label{def degenerate pseudo arc}
A pseudo arc $ \mathcal{H} = ( \mathcal{H}_i )_{i=1}^n $ of type $ [n,r,m,t] $ is nondegenerate if $ \sum_{i=1}^n \mathcal{H}_i = \mathbb{F}_q^m $. 
\end{definition}

We will now define a correspondence between pseudo arcs and linear codes in the folded Hamming distance.

\begin{definition} \label{def corresp pseudo arcs}
Let $ \mathcal{H} = ( \mathcal{H}_i )_{i=1}^n $ be a pseudo arc of type $ [n,r,m,t] $. We say that $ \mathbf{h} = (\mathbf{h}_{i,j})_{i=1,j=1}^{n,r} $ is a basis of $ \mathcal{H} $ if $ \mathbf{h}_{i,1},\ldots, \mathbf{h}_{i,r} \in \mathbb{F}_q^m $ are $ m \times 1 $ column vectors forming a basis of $ \mathcal{H}_i $, for $ i \in [n] $. Next, define the $ m \times (rn) $ matrix
$$ H_\mathbf{h} = (\mathbf{h}_{1,1}, \ldots, \mathbf{h}_{1,r} | \ldots | \mathbf{h}_{n,1} , \ldots , \mathbf{h}_{n,r}). $$
Finally, we define the linear code $ \mathcal{C}_\mathbf{h} \subseteq \mathbb{F}_q^{rn} $ as that with parity-check matrix $ H_\mathbf{h} $.

Conversely, given a linear code $ \mathcal{C} \subseteq \mathbb{F}_q^{rn} $ of dimension $ k = rn - m $ with a (full rank) parity-check matrix $ H $, if $ \mathbf{h}_{i,j} \in \mathbb{F}_q^{m \times 1} $ is the $ ((i-1)r + j) $th column of $ H $, then we define $ \mathcal{H}_H = (\mathcal{H}_i)_{i=1}^n $, where $ \mathcal{H}_i $ is the subspace generated by $ \mathbf{h}_{i,1}, \ldots, \mathbf{h}_{i,r} $, for $ i \in [n] $.
\end{definition}

We have the following exact correspondence between parameters.

\begin{theorem} \label{th parameters pseudo arcs}
\begin{enumerate}
\item
If $ \mathcal{H} = (\mathcal{H}_i)_{i=1}^n $ is a nondegenerate pseudo arc of type $ [n,r,m,t] $ with basis $ \mathbf{h} $, then $ \mathcal{C}_\mathbf{h} $ is a linear code of type $ [n,r,k,d] $ with $ k = rn - m $ and $ d = t+1 $.
\item
If $ \mathcal{C} $ is a linear code of type $ [n,r,k,d] $ and $ H $ is one of its parity-check matrices, then $ \mathcal{H}_H $ is a nondegenerate pseudo arc of type $ [n,r,m,t] $ with $ m = rn - k $ and $ t = d-1 $.
\end{enumerate}
\end{theorem}
\begin{proof}
Item 2 is proven similarly, thus we only prove Item 1. The fact that $ \mathcal{H} $ is nondegenerate is equivalent to $ H_\mathbf{h} $ having rank $ m $. Thus $ k = \dim(\mathcal{C}_\mathbf{h}) = rn - m $. Finally $ d = t+1 $ follows by combining Proposition \ref{prop d from generator and parity matrices} (Item 2) and Definition \ref{def pseudo arcs} (Item 2).
\end{proof}

\begin{remark}
A similar connection is made in \cite{liu}. However, the notion of degenerateness of pseudo arcs is not considered there. In particular, the relations between the corresponding parameters could only be given in \cite{liu} as bounds (which are not always tight) instead of as exact equalities.
\end{remark}

However the correspondence in Definition \ref{def corresp pseudo arcs} is not bijective. In fact, for every pseudo arc, we may obtain multiple codes, and viceversa. In order to obtain a bijection, we need to consider equivalent codes (Definition \ref{def equivalence}) and equivalent pseudo arcs, which we now define.

\begin{definition} \label{def equivalence pseudo arcs}
Given pseudo arcs $ \mathcal{H} = (\mathcal{H}_i)_{i=1}^n $ and $ \mathcal{H}^\prime = (\mathcal{H}^\prime_i )_{i=1}^n $, both of type $ [n,r,m,t] $, an equivalence between them is a pair $ (\varphi, \sigma) $, where $ \sigma : [n] \longrightarrow [n] $ is a permutation and $ \varphi : \mathbb{F}_q^m \longrightarrow \mathbb{F}_q^m $ is a vector space isomorphism such that $ \mathcal{H}^\prime_{\sigma(i)} = \varphi(\mathcal{H}_i) $, for all $ i \in [n] $. If it exists, we say that $ \mathcal{H} $ and $ \mathcal{H}^\prime $ are equivalent (and clearly one is degenerate if, and only if, so is the other).
\end{definition}

We may now obtain a bijection between equivalence classes of linear codes (Definition \ref{def equivalence}) and pseudo arcs (Definition \ref{def equivalence pseudo arcs}). 

\begin{theorem} \label{th bijection pseudo arcs}
\begin{enumerate}
\item
Let $ \mathcal{H} $ and $ \mathcal{H}^\prime $ be equivalent nondegenerate pseudo arcs with bases $ \mathbf{h} $ and $ \mathbf{h}^\prime $, respectively. Then $ \mathcal{C}_\mathbf{h} $ and $ \mathcal{C}_{\mathbf{h}^\prime} $ are equivalent.
\item
Let $ \mathcal{C} $ and $ \mathcal{C}^\prime $ be equivalent linear codes in $ \mathbb{F}_q^{rn} $ with parity-check matrices $ H $ and $ H^\prime $, respectively. Then $ \mathcal{H}_H $ and $ \mathcal{H}_{H^\prime} $ are equivalent.
\end{enumerate}
\end{theorem}
\begin{proof}
Item 2 is proven similarly, thus we only prove Item 1. Let $ \mathcal{H} = (\mathcal{H}_i)_{i=1}^n $ and $ \mathcal{H}^\prime = (\mathcal{H}^\prime_i)_{i=1}^n $, both of type $ [n,r,m,t] $, and let $ (\varphi,\sigma) $ be the equivalence between them. There exists an invertible matrix $ B \in {\rm GL}_m(\mathbb{F}_q) $ such that $ \varphi(\mathbf{x}) = B \mathbf{x} $, for every $ m \times 1 $ column vector $ \mathbf{x} \in \mathbb{F}_q^m $. Now the equality $ \mathcal{H}^\prime_{\sigma(i)} = \varphi(\mathcal{H}_i) $ means that there exists an invertible matrix $ A_i \in {\rm GL}_r(\mathbb{F}_q) $ such that 
\begin{equation}
B (\mathbf{h}_{i,1}, \ldots, \mathbf{h}_{i,r}) = (\mathbf{h}^\prime_{\sigma(i),1}, \ldots, \mathbf{h}^\prime_{\sigma(i),r}) A_i,
\label{eq proof bijection pseudo arcs}
\end{equation}
for every $ i \in [n] $. Let $ P_\sigma $ be the only $ (rn) \times (rn) $ matrix over $ \mathbb{F}_q $ such that $ (\mathbf{c}_1, \ldots, \mathbf{c}_n) P_\sigma = (\mathbf{c}_{\sigma(1)}, \ldots, \mathbf{c}_{\sigma(n)}) $, for all $ \mathbf{c}_1, \ldots, \mathbf{c}_n \in \mathbb{F}_q^r $. Then (\ref{eq proof bijection pseudo arcs}) is equivalent to
$$ B H_\mathbf{h} = H_{\mathbf{h}^\prime} P_\sigma {\rm Diag}(A_1, \ldots, A_n), $$
where $ {\rm Diag}(A_1, \ldots, A_n) $ denotes the block diagonal matrix with $ A_1, \ldots, A_n $ in the main diagonal. Thus we deduce that $ \mathcal{C}_\mathbf{h}^\perp $ and $ \mathcal{C}_{\mathbf{h}^\prime}^\perp $ are equivalent by Proposition \ref{prop isometries}. Hence $ \mathcal{C}_\mathbf{h} $ and $ \mathcal{C}_{\mathbf{h}^\prime} $ are equivalent by Corollary \ref{cor dual equivalence}, and we are done.
\end{proof}

Combining Theorems \ref{th parameters pseudo arcs} and \ref{th bijection pseudo arcs}, we conclude the following.

\begin{corollary}
The correspondence in Definition \ref{def corresp pseudo arcs} induces a bijection between equivalent classes of linear codes of type $ [n,r,k,d] $ and equivalent classes of nondegenerate pseudo arcs of type $ [n,r,m,t] $ with $ m = rn - k $ and $ t = d-1 $.
\end{corollary}

As a consequence, all of the results in this paper concerning linear codes in the folded Hamming distance can be immediately translated to results for nondegenerate pseudo arcs. We leave the details to the reader.

\section*{Acknowledgements}

This work has been supported by MCIN/AEI/10.13039/501100011033 and the European Union NextGenerationEU/PRTR (Grant no. TED2021-130358B-I00), and by MICIU/AEI/ 10.13039/501100011033 and ERDF/EU (Grant no. PID2022-138906NB-C21).

\bibliographystyle{plain}

\end{document}